\documentclass[oneside,12pt]{article}
\usepackage{geometry}
\usepackage{subcaption}
\usepackage{amssymb}
\usepackage{amsmath}
\usepackage{amsthm}

\newtheorem{theorem}{Theorem}
\newtheorem{Lemma}[theorem]{Lemma}
\newtheorem{Proposition}[theorem]{Proposition}

\newtheorem{remark}[theorem]{Remark}
\newtheorem{examples}[theorem]{Example}

\theoremstyle{definition}
\newtheorem{Definition}[theorem]{Definition}
\newtheorem{problem}{Problem}

\usepackage{mathtools}

\usepackage{enumitem}

\newcommand{\al}{\alpha}

\newcommand{\eps}{\epsilon}

\newcommand{\sr}{\mathcal S}

\newcommand{\R}{\mathbb R}

\newcommand{\N}{\mathbb N}

\newcommand{\Fo}{\mathbf{F}}
\newcommand{\Uo}{\mathbf{U}}
\newcommand{\Ro}{\mathbf{R}}
\newcommand{\Rd}{\mathrm R}

\newcommand{\reg}{r}

\newcommand\sabs[1]{{\lvert#1\rvert}}
\newcommand\norm[1]{\Vert#1\Vert}
\newcommand\set[1]{{\{#1\}}}

\newcommand\snorm[1]{\Vert#1\Vert}

\DeclareMathOperator*{\argmin}{arg\,min}

\DeclarePairedDelimiter{\abs}{\lvert}{\rvert}

\newcommand{\XX}{\mathbb X}

\newcommand{\sph}{\mathbb{S}}

\newcommand{\signal}{f}
\newcommand{\data}{y}
\newcommand{\ph}{\varphi}

\numberwithin{equation}{section}
\numberwithin{figure}{section}
\numberwithin{theorem}{section}

\newcommand{\astx}{\circledast}
\newcommand{\asts}{\circledast_s}
\title{Feature reconstruction from incomplete tomographic data without detour}

\author{Simon G\"oppel \thanks{Department of Mathematics, University of Innsbruck,
Technikerstrasse 13, 6020 Innsbruck, Austria, \{simon.goeppel, markus.haltmeier\}@uibk.ac.at} \and J\"urgen Frikel \thanks{Department of Computer Science and Mathematics, Galgenbergstra{\ss}e 32, D-93053 Regensburg, Germany, juergen.frikel@oth-regensburg.de}\and Markus Haltmeier\footnotemark[1]  }

\begin{document}

\maketitle

\begin{abstract}
In this paper, we consider the problem of feature reconstruction from incomplete x-ray CT data. Such problems occurs, e.g., as a result of dose reduction in the context medical imaging. Since image reconstruction from incomplete data is a severely ill-posed problem, the reconstructed images may suffer from characteristic artefacts or missing features, and significantly complicate subsequent image processing tasks (e.g., edge detection or segmentation). In this paper, we introduce a novel framework for the robust reconstruction of convolutional image features directly from CT data, without the need of computing a reconstruction firs. Within our framework we use non-linear (variational) regularization methods that can be adapted to a variety of feature reconstruction tasks and to several limited data situations . In our numerical experiments, we consider several instances of edge reconstructions from angularly undersampled data and show that our approach is able to reliably reconstruct feature maps in this case.

\noindent\textbf{keywords}
Computed tomography; Radon transform; reconstruction, limited data; sparse data, feature reconstruction, edge detection.

\end{abstract}

\section{Introduction}
\label{sec:intro}

Computed tomography (CT) has established itself as one of the standard diagnostic tools in medical imaging. However, the relatively high radiation dose that is used to produce high resolution CT images (and that patients are exposed to) has become one of the major clinical concerns \cite{CTDoseReduction,cancer,thousands,middleage}. The reduction of the radiation exposure of a patient while ensuring the diagnostic image quality constitutes one of the main challenges in CT. Beside patient safety, also the reduction of scanning times and costs  constitute important aspects of dose reduction, which is often achieved by reducing the x-ray energy level (leading to higher noise levels in the data) or by reducing the number of collected CT data (leading to incomplete data), cf. \cite{CTDoseReduction}. 

In this work, we particularly consider incomplete data situations, e.g., that arise in a sparse or limited view setup, where CT data is collected only with respect to a small number of x-rays directions or within a small angular range. The intentional reduction of the angular sampling rate leads to an under-determined and severely ill-posed image reconstruction problem, c.f. \cite{natterer}. As a consequence, the reconstructed image quality can be substantially degraded, e.g., by artefacts or missing features \cite{juergentodd}, and complicate subsequent image processing tasks (such as edge detection or segmentation) that are often employed within a CAD-pipeline (computer aided diagnosis). Therefore, the development of robust feature detection algorithms for CT that ensure the diagnostic image quality is an important and very challenging task. In this paper, we introduce a framework for feature reconstruction directly from incomplete tomographic data, which is in contrast to the classical 2-step approach where reconstruction and feature detection are performed in two separate steps.

\subsection*{Incomplete tomographic data}

In this article, we consider the parallel beam geometry and use the 2D Radon transform $\Ro f\colon \sph^1 \times \R \to \R$ as a model for the (full) CT data generation process, where $\sph^1$ denotes the unit sphere in $\R^2$ and $f \colon \R^2\to \R$ is a function representing the sought tomographic image (CT scan). Here, the value $\Ro f (\theta, s)$ represents one x-ray measurement over a line in $\R^2$ that is parametrized by the normal vector $\theta \in \sph^1$ and the signed distance from the origin $s \in \R$. In what follows, we consider incomplete data situations where the Radon data are available only for a small number of directions, given by $\Theta  \coloneqq \{ \theta_1, \dots, \theta_m \}$. We denote the angularly sampled tomographic Radon data by $\Ro_\Theta  f \coloneqq (\Ro  f)|_{\Theta \times \R}$. In this context, the (semi-discrete) CT data $\Ro_\Theta  f$ will be called incomplete if the Radon transform is insufficiently sampled with respect to the directional variable. Prominent instances of incomplete data situations are: \emph{sparse angle setup}, where the directions in $\Theta$ are sparsely distributed over the full angular range  $[0,\pi]$; \emph{limited view setup}, where $\Theta$ covers only small part of the full angular range $[0,\pi]$. Precise mathematical criteria of (in-) sufficient sampling can be derived from the Shannon sampling theory. Those criteria are based on the relation between the number of directions $m=|\Theta|$ and the bandwidth of $f$, cf. \cite{natterer}. In this work, we will mainly focus on the sparse angle case, with uniformly distributed directions $\theta_1, \dots, \theta_m$ on a half-sphere, e.g., directions $\theta_k\coloneqq\theta(\varphi_k)=(\cos(\varphi_k),\sin(\varphi_k))^\top$ with uniformly distributed angles  $\varphi_k\in[0, \pi]$.

\subsection*{Feature reconstruction in tomography}

In the following, we consider image features that can be extracted from a CT scan $\signal \in L^2(\R^2)$ by a convolution with a kernel $U\in L^1(\R^2)$. In this context, the notion of a feature map will refer to the convolution product $ \signal  \astx U$, and the convolution kernel $U$ will be called feature extraction filter. Examples of feature detection tasks that can be realized by a convolution include edge  detection, image restoration, image enhancement, or texture filtering \cite{jain1989fundamentals}. For example, in case of edge detection, the filter $U$ can be chosen as a smooth approximation of differential operators, e.g., of the Laplacian operator \cite{jahne2005digital}. In our practical examples, we will mainly focus on edge detection in tomography. However, the proposed framework also applies to more general feature extraction tasks. 

In many standard imaging setups, image reconstruction and feature extraction are realized in two separate steps. However, as pointed out in \cite{louis2}, this 2-step approach can lead to unreliable feature maps since feature extraction algorithms have to account for inaccuracies that are present in the reconstruction. This is particularly true for the case of incomplete CT data as those reconstructions may contain artefacts. Hence, combining these two steps into an approach that computes feature maps directly from CT data can lead to a significant performance increase, as was already pointed out in \cite{louis1,louis2}. In this work, we account for this fact and extend the results of \cite{louis1,louis2} to more general setting and, in particular, to limited data situations.

\subsection*{Main contributions and related work}

In this work, we propose a framework  to directly reconstruct the feature map $  U \astx \signal$ from  the measured tomographic  data. Our approach is based on the  forward convolution identity for the Radon transform, that is $\Ro (\signal \astx U) = (\Ro  \signal) \asts (\Ro  U)$, where on the right hand side the convolution is taken with respect to the second variable of the Radon transform, cf. \cite{natterer}. This identity implies that, given (semi-discrete) CT data, the feature map satisfies the (discretized) equation $\Ro_\Theta h = \data_\Theta$, where $\data_\Theta=\Ro_\Theta f\asts\Ro_\Theta U$ is the modified (preprocessed) CT data. Therefore, the sought feature map can be formally computed by applying a discretized version of the inverse Radon transform to $\data_\Theta$, i.e., as $h_\Theta = \Ro_\Theta^{-1} (\data_\Theta)$. In case of full data (sufficient sampling), this can be accurately and efficiently computed by using the well known filtered backprojection (FBP) algorithm with the filter $\Ro_\Theta  U$. However, if the CT data is incomplete,  this approach would lead to unreliable feature maps since in such situations the FBP is known to produce inaccurate reconstruction results, cf. \cite{natterer,juergentodd}. 

In order to account for data incompleteness, we propose to replace the inverse $\Ro_\Theta^{-1}$ by a suitable regularization method for $\Ro_\Theta^{-1}$ that is also able to deal with undersampled data. More concretely, we propose to reconstruct the (discrete) feature map $h_\Theta$ by minimizing the following Tikhonov type functional: 
\[ 
    h_\Theta \in \argmin_h \frac{1}{2} \, \norm{\Ro_\Theta h - u_\Theta  \asts  \data_\Theta }^2 +  \reg( h ) \,. \]
This framework offers a flexible way for incorporating a-priori information about the feature map into the reconstruction and, in this way, to account for the missing data. For example, from
the theory of compressed sensing it is well known that sparsity can help to overcome the classical Nyquist-Shannon-Whittaker-Kotelnikov paradigm \cite{cs}. Hence, whenever the sought feature map is known to be sparse (e.g., in case of edge detection), sparse regularization techniques can be easily incorporated into this framework.  

Approaches that combine image reconstruction and edge detection have been proposed for the case of full tomographic data, e.g., in \cite{louis1,louis2}. Although the presented work follows the spirit of \cite{louis1,louis2}, it comes with several novelties and advantages. On a formal level, our approach is based on the forward convolution identity, in contrast to the dual convolution identity, given by  $(\Ro^\ast u) \astx \signal = \Ro^\ast ( u \asts \Ro \signal)$, that is employed in in \cite{louis1,louis2}. The latter requires full  (properly sampled) data, since the backprojection operator $\Ro^\ast$ integrates over the full angular range (requiring proper sampling in the angular variable). In contrast, our framework is applicable to incomplete Radon data situations, since the forward convolutional identity (used in our approach) can be applied to more general situations. Moreover, in  order to recover the feature map $U \astx \signal$, we use non-linear regularization methods that can be adapted to a variety of situations and incorporate different kinds of prior information. From this perspective, our approach also offers more flexibility. A similar approach was presented in our recent proceedings article \cite{frikel2021combining}, where the main focus was on the stable recovery of the image gradient from CT data and its application to Canny edge detection. Following the ideas of \cite{louis1,louis2} similar feature detection methods were developed also for other types of tomography problems, e.g., in \cite{Hahn_2013,Rigaud_2015,Rigaud2017}. Besides that, we are not aware of any further results concerning convolutional feature reconstruction from incomplete x-ray CT data. 

Combinations of reconstruction and segmentation have also been presented in the literature for different types of tomography problems, e.g., in \cite{elangovan2001sinograms,klann2011mumford,storath2015joint,Burger_2016,Romanov2016,LQWJ2018,JoinRecSeg2018}. As a commonality to our approach, many of those methods are based on the minimization of an energy functional of the form $\norm{\Ro_\Theta \signal - \data}^2 +  r(\signal\ast \Uo)$, incorporating feature maps as prior information. Important examples include Mumford-Shah like  approaches  \cite{klann2011mumford,Burger_2016,JoinRecSeg2018,LQWJ2018} or the Potts model \cite{storath2015joint}.  Also,  geometric approaches for computing segmentation  masks directly from tomographic data were employed in \cite{elangovan2001sinograms}.

\subsection*{Outline}

In Section \ref{sec:background} we provide some basic facts about the Radon transform, sampling and sparse recovery. In Section \ref{sec:feature} we introduce the proposed feature reconstruction framework and present several examples of convolutional feature reconstruction filters along with corresponding data filters, mainly focusing on the case of edge detection. Experimental results will be presented in Section \ref{sec:num}. We conclude with a summary and outlook given  in Section \ref{sec:outlook}.

\section{Materials and Methods}
\label{sec:background}

In this section, we recall some basic facts about the 2D Radon transform, including important identities and sampling conditions. In particular, we define the sub-sampled Radon transform that will be used throughout this article. Although, our presentation is restricted to the 2D case (because this makes the presentations more concise and clear), the presented concepts can be easily generalized to the $d$-dimensional setup.

\subsection{The Radon transform}

Let  $\sr(\R^2)$ denote the Schwartz space on $\R^2$ (space of smooth  functions that are rapidly decaying together with all their derivatives) and $\sr(\sph^{1} \times \R)$ the Schwartz space over $\sph^{1} \times \R$ as the space of all smooth functions that are rapidly decaying together with all their derivatives in the second component, cf. \cite{natterer}. We consider the Radon transform as an operator between those Schwartz spaces, $\Ro \colon \sr(\R^2) \to \sr(\sph^{1} \times \R)$, which is defined via

\begin{equation}\label{eq:radon}
   \Ro f(\theta, s) \coloneqq  \int_{-\infty}^\infty f(s \theta + t\theta^\bot) \mathrm{d} t,
\end{equation}

where $s\in\R$, $\theta\in\sph^1$ and $\theta^\bot$ denotes the rotated version of $\theta$ by $\pi/2$ counterclockwise (in particular, $\theta^\bot$ is a unit vector  perpendicular to $\theta$). The value $\Ro f(\theta, s)$ represents one x-ray measurement along the x-ray path that is given by the line $L(\theta,s)=\{x\in\R^2:\langle x,\theta\rangle=s\}$. Since $L(-\theta,-s)=L(\theta,s)$, the following symmetry property holds for the Radon transform, $\Ro f(-\phi,-s)=\Ro f(\theta,s)$. Hence, it is sufficient to know the values of Radon transform only on a half sphere, e.g., on the upper half sphere. Such data is therefore considered to be complete. The dual transform (backprojection operator) is defined as $ \Ro^\ast:\sr(\sph^{1} \times \R) \to \sr(\R^2)$, 
\begin{equation}\label{eq:backprojection}
   \Ro^\ast g(x) \coloneqq  \int_{\sph^1} g(\theta,\theta\cdot x) \mathrm{d} \theta.
\end{equation}

The  Radon transform is a well defined linear and injective operator, and several  analytic properties are well-known. One of the most important properties is the so-called Fourier slice theorem that describes the relation between the Radon and the Fourier transforms. In order to state this relation, we first recall that the Fourier transform is defined as $\Fo: \sr(\R^d)\to\sr(\R^d)$,  $\Fo \signal ( \xi ) \coloneqq  (2\pi)^{-d/2}=\int_{\R^d } \signal( x) e^{-  i x \cdot \xi} \,  \mathrm{d} x $ for $d\in\N$. Whenever convenient, we will also use the abbreviated notation $\hat f(\xi):=\Fo f(\xi)$. The Fourier transform is a linear isomorphism on the Schwartz space $\sr(\R^d)$ and its inverse is given by $\check f(x):=\Fo^{-1} \signal ( x ) =  (2\pi)^{-d/2}\int_{\R^2 } \signal(\xi) e^{ i x \cdot \xi} \,  \mathrm{d} \xi $. In what follows, we will denote the convolution of two functions $\signal, g:\R^d\to\R$ by $\signal \astx g (x) := \int_{\R^d} f(x-y)g(y) \mathrm{d} y$, where $d\in\N$. Moreover, for functions $g\in\sr(\sph^1\times \R)$, the Fourier transform $\Fo_s g$ will refer to the 1D-Fourier transform of $g$ with respect to the second variable. Analogeously, $g\asts h$ will denote the convolution of $g,h:\sph^1\times \R\to \R$ with respect to the second variable.

\begin{Lemma}[Properties of the Radon transform] \label{lem:radon}\mbox{}
\begin{enumerate}[label=(R\arabic*), leftmargin=3em]
\item\label{r1} \emph{Fourier slice theorem:} $\forall \signal \in \sr(\R^2) \; \forall (\theta, s) \in  \sph^{1} \times \R \colon \Fo_s \Ro f (\theta, \sigma) = \sqrt{2\pi} \cdot\Fo f (\theta \sigma)$.

\item\label{r2}  \emph{Convolution identity:} $\forall U,\signal \in \sr(\R^2) \colon \Ro (\signal \astx U) = \Ro \signal \asts \Ro U$.

\item\label{r3}  \emph{Dual convolution identity:} 
$\forall u \in \sr(\sph^{1} \times \R) \; \forall \signal  \in \sr(\R^2)    \colon 
\Ro^* u \astx f = \Ro^\ast (u \asts \Ro f)$.

\item\label{r4}  \emph{Intertwining with Derivatives:} 
$\forall \al \in \N^2 \; \forall \signal  \in \sr(\R^2) 
\colon  \Ro \partial_x^\alpha  \signal = \theta^\alpha 
\partial_s^{|\alpha|} \Ro   \signal$

\item\label{r5}  \emph{Intertwining with Laplacian:}
$ \forall \signal  \in \sr(\R^2) 
\colon  \Ro \Delta_x  \signal =  
\partial_s^2 \Ro   \signal$.
\end{enumerate}
\end{Lemma}

\begin{proof}
All identities are derived in  \cite[Chapter II]{natterer}.
\end{proof}

The  approach that we are going to present in Section \ref{sec:feature} is based on the convolution identity and can be formulated for arbitrary  spatial dimension $d \geq 2$. For the sake of clarity we consider two spatial dimensions $d=2$ in this paper. In this case, we will use the parametrization of $\sph^1$ given by $\theta(\varphi) \coloneqq  (\cos(\varphi),\sin(\varphi))^\top$ with $\varphi \in [0,\pi)$. Then $\theta^\bot(\varphi)=(-\sin(\varphi),\cos(\varphi)^\top$.  For the Radon transform we will  (with some abuse of notation) write   \[\Ro f (\varphi, s):=\Ro f (\theta(\varphi), s).\]


\subsection{Sampling the Radon transform}

Since in practice one has to deal with discrete data, we are forced to work with discretized (sampled) versions of the Radon transform. In this context, questions about proper sampling arise. In order to understand what it means for the CT data to be complete (properly sampled) or incomplete (improperly sampled), we recall some basic facts from the Shannon sampling theory for the Radon transform for the case of parallel scanning geometry (see for example \cite[Section III]{natterer}).

In what follows, we assume that $\signal$ is compactly supported on the unit disc $D\subseteq\R^2$ and consider sampled CT data $\Ro f(\varphi_j,s_l)$ with $N_\ph\in\N$ equispaced angles $\varphi_j$ in $[0,\pi)$ and $N_s$ equispaced values $s_l$ in $[-1,1]$ for the $s$-variable, i.e.,
\begin{equation}\label{eq:sampling}
(\varphi_j, s_\ell) 
= \biggl( \frac{j\pi}{ N_\ph}, \frac{\ell}{N_s}\biggr) \quad 
 \text{ for }    (j, \ell) \in \{0,\ldots,  N_\ph-1 \} \times   \{ - N_s,\ldots, N_s \} 
\,.
\end{equation}
For  given sampling points  \eqref{eq:sampling}  and a finite dimensional subspace $\XX_0\subseteq  \sr (\R^d)$ we define the \emph{discrete Radon transform} as
\begin{equation}\label{eq:radon-discrete}
\Rd \colon \XX_0 \to  \R^{N_\ph  \times (2 N_s+1)}
\colon  \signal \mapsto  (  \Ro \signal (\theta_j,s_\ell) )_{j,\ell} \,.
\end{equation}
The basic question of classical sampling theory in the context of CT is to find  conditions on the class of images $\signal \in \XX_0$ and on the sampling points under which the sampled data $\Rd \signal$ uniquely determines the unknown function $\signal$.  Sampling theory for  CT has been studied, for example,  in \cite{desbat1993efficient,Far04,Far06,natterer95sampling,RatLin81}.
While the classical sampling theory (e.g., in the setting of classical signal processing) works with the class of band-limited functions, the sampling conditions in the context of CT are typically derived for the class of essentially band-limited function.

\begin{remark}[Band-limited and essentially band-limited functions]
A function $\signal \in L^2(\R^2)$ is called $b$-band-limited if  its Fourier transform  $\Fo \signal (\xi)$ vanishes for $\norm{\xi} >b$. A function $\signal$ is called essentially $b$-band-limited if $\hat{f}(\xi)$ is negligible for $\norm{\xi} \geq b$ in the sense that
$ \eps_0(f, b) \coloneqq \int_{\snorm{\xi} \geq b} \sabs {\Fo f(\xi)} d\xi$
is sufficiently small, see \cite{natterer}.  One reason for working with essentially band-limited functions in CT is that  functions with compact support cannot be strictly band-limited. However, the quantity $\eps_0(f, b)$ can become arbitrarily small  for functions with compact support.  
\end{remark}

The bandwidth $b$ is crucial for the correct sampling conditions and the calculation of appropriate filters.  If  $\XX_0$ consists of essentially $b$-band-limited functions that vanish outside the unit disc $D$, then the correct sampling conditions are given by \cite{natterer}
\begin{equation}\label{eq:sampling-cond}
	(N_\ph, N_s)  
	\coloneqq \bigl( \lceil  b \rceil , \lceil    b  /  \pi \rceil \bigr) \,.
\end{equation}
Obviously, as the bandwidth $b$ increases, the step sizes $\pi / N_\ph$ and $1/N_s$ have to decrease in order that \eqref{eq:sampling-cond} is  satisfied.  Thus, if the bandwidth $b$ is large, a large number measurements (roughly $2 b^2  / \pi$) have to be collected. As a consequence, for high resolution imaging the sampling conditions require a large number of measurements. Thus, in practical applications, high resolution imaging in CT also leads to large scanning times and to high doses of x-ray exposure. A classical approach for dose reduction consists in reduction of x-ray measurements. For example, this can be achieved by angular undersampling where Radon data is collected only for a relatively small number of directions $\Theta\subseteq \set{\theta_0, \dots, \theta_{N_\ph-1}}$.

\begin{Definition}[Sub-sampled Radon transform]
Let  $(N_\ph, N_s)$ be defied by  \eqref{eq:sampling-cond} and  let $\XX_0$  be the set of essentially $b$-band-limited functions that vanishes outside the unit disc $D$ (note that in that case the discrete Radon transform defined in \eqref{eq:radon-discrete} is correctly sampled). For $ \Theta \subseteq \set{\theta_0, \dots, \theta_{N_\ph-1}}$  we call
\begin{equation}\label{eq:radon-sub}
 \Rd_\Theta \colon \XX_0 \to  \R^{\abs{\Theta}  \times (2 N_s+1)}
\colon   f    \mapsto (\Rd f) |_{\Theta \times \{-N_s, \dots, N_s\} }
\end{equation}
the sub-sampled discrete Radon transform. We will also use the semi-discrete  Radon transform $\Ro_\Theta f : = (\Ro f)|_{\Theta  \times \R}$, where we only sample in the angular direction but not in the radial direction.
\end{Definition}

If we  perform actual undersampling, where the number of directions in $\Theta$ is much less than $N_\ph$,  then the linear equation  $ \Rd_\Theta \signal =  \data_\Theta $ will be is severely under-determined and its solution requires additional  prior information (e.g., sparsity of the feature map). 

 
\section{Feature reconstruction from incomplete data}
\label{sec:feature}

In this section, we present our approach for feature map reconstruction from incomplete data. For a given bandwidth $b$, we let $\XX_0$ denote the set of essentially $b$-band-limited functions that vanishes outside $D$. Furthermore, we assume that the set of directions $\set{\theta_0, \dots, \theta_{N_\ph-1}}$ is chosen according to the sampling conditions \eqref{eq:sampling-cond}. 

\begin{problem}[Feature reconstruction task] \label{pr:feature}
Let $ \Theta \subseteq \set{\theta_0, \dots, \theta_{N_\ph-1}}$ and let $\data_\Theta:\Theta\times\R \to \R$ be the noisy subsampled (semi-discrete) CT data with $\norm{\Ro_\Theta \signal - \data_\Theta} \leq \delta$, where $\signal \in \XX_0$ is the true but unknown image and  $\delta > 0$ the known noise level. Given a feature extraction filter $U \colon \R^2 \to \R$,  our goal is to estimate the feature map $U \astx \signal$ from the (undersampled) data $y_\Theta$.
\end{problem}

\begin{remark}
\begin{enumerate}[itemsep=1ex]
\item[]
    \item From a general perspective, the Problem \ref{pr:feature} is related to the field of optimal recovery \cite{micchelli1977survey} where the goal is to estimate certain features of an element in a space $\XX_0$ from noisy indirect observations.
    
    \item Depending on the particular choice of the filter $U$, the Problem \ref{pr:feature} corresponds to several typical tasks in tomography. For example, if $U$ is chosen as an approximation of the Delta distribution, the Problem \ref{pr:feature} is equivalent to the classical image reconstruction problem. In fact, the filtered backprojection algorithm (FBP) is derived in this way from  the dual convolution identity \ref{r3} for the full data case, cf. \cite{natterer}. Another instance of Problem \ref{pr:feature} is edge reconstruction from tomographic data $y_\Theta$. For example, this can be achieved by choosing the feature extraction filter $U$ as the Laplacian of an approximation to the Delta distribution (e.g., Laplacian of Gaussian (LoG)). Then, the Problem \ref{pr:feature} boils down to an approximate recovery of the Laplacian of $f$, which is used in practical edge-detection algorithms (e.g., LoG-filter \cite{jain1989fundamentals,jahne2005digital}).
    
    \item Traditionally, the solution of Problem \ref{pr:feature} is realized via the 2-step approach: First, estimate $\signal$ and, second, apply convolution in order to estimate the feature map $U  \astx \signal $. This 2-step approach has several disadvantages: Since image reconstruction in CT is (possibly severely) ill-posed, the fist step might introduce huge errors in the reconstructed image. Those errors will also be propagated through the second (feature extraction) step that itself can be ill-posed and even further amplify errors. In order to reduce the error propagation of the first step, regularization strategies are usually applied. The choice of a suitable regularization strategy strongly depends on the particular situation and on the available prior information about the sought object $\signal$. However, the recovery of $\signal$ requires different prior knowledge than feature extraction. This dismatch can lead to a substantial loss of performance in the feature detection step. 
    
    \item In order to overcome the limitations mentioned in the remark above, image reconstruction and edge detection were combined in \cite{louis2,louis1}, where explicit formulas for estimating the edge map have been derived using the method of approximate inverse. This approach is also based on the dual convolution identity \ref{r3} and is closely related to the standard filtered backprojection (FBP) algorithm. However, this approach is not applicable to the case of undersampled data, since \cite{louis1,louis2} employ the dual convolutional identity \ref{r3} and calculate the reconstruction filters of the form $\Ro^*_\Theta u$. In this calculation, in order to achieve a good approximation of the integral in \ref{eq:backprojection}, a properly sampled Radon data is required. 
\end{enumerate}
\end{remark}

To overcome the limitations mentioned in the last remark above, we derive a novel framework for feature reconstruction in the next subsection (based on the forward convolutional identity \ref{r3}) that does not make use of  the continuous backprojection and, hence, can be applied to more general situations.

\subsection{Proposed feature reconstruction}

Our proposed framework for solving the feature reconstruction Problem~\ref{pr:feature} is based on the forward convolution identity \ref{r2} stated in Lemma~\ref{lem:radon}. Because the convolution on the right-hand side of \ref{r2} acts only on the second variable, the convolution identity is not affected by the subsampling in the  angular direction. Therefore, we have
\begin{equation}\label{eq:conv}
	\Ro_\Theta ( \signal \astx U ) = u_\Theta  \asts \Ro_\Theta\signal  
	\quad 
	\text{ with  } u_\Theta \coloneqq \Ro_\Theta U \,.
\end{equation}

Formally, the solution of \eqref{eq:conv} takes the form $\signal \astx U=\Ro_\Theta^{-1}(u_\Theta  \asts \Ro_\Theta\signal)$. If the data is properly sampled, this can be accurately and efficiently computed by applying the FBP algorithm to the filtered CT data $y_\Theta= u_\Theta  \asts \Ro_\Theta\signal$. In this context, the data filter $u_\Theta$ needs to be precomputed (from a given feature extraction filter $U$) in a filter design step. However, if the data $\Ro_\Theta \signal$ is not properly sampled, the equations \eqref{eq:conv} are underdetermined and, in this case, FBP doesn't produce accurate results, cf. \cite{natterer,juergentodd}. In order to account for data incompleteness and stably approximate the feature map $\signal \astx U$, a-priori information about the specific feature kernel $U$ or the feature map $f\astx U$ needs to be integrated into the reconstruction procedure. As a flexible way for doing this, we propose to approximate the inverse $\Ro_\Theta^{-1}$ by the following variational regularization scheme: 

\begin{equation}\label{eq:feature-rec}
	\frac{1}{2} \norm{\Ro_\Theta h - u_\Theta  \asts  \data_\Theta }_2^2
	+  \reg( h ) 
	\to 
	\min_{h \in \XX_0}    \,.
\end{equation}
Here  $\data_\Theta:\Theta  \times \R\to\R $ denotes the noisy (semi-discrete data) and $\reg \colon \XX_0 \to [0, \infty]$ is a regularization (penalty) term.

\begin{examples}
\label{ex:problems}
\begin{enumerate}[itemsep=1ex]
\item[]
 \item \textsc{Image reconstruction:}
Here, the feature extraction filter $U=U_\alpha$ is chosen as an approximation to the Delta distribution. For example, as $U = g_\alpha$ with 
\begin{equation}
\label{eq:gaussian2}
    g_\alpha(x) = \frac{1}{2\pi\alpha^2} \exp \left(-\frac{\norm{x}^2}{2\alpha^2}\right),\quad \alpha>0
\end{equation} 
being the Gaussian kernel. Another way of choosing $U$ for reconstruction purposes is through ideal low-pass filters $U_\alpha$, that are defined in the Frequency domain via $\Fo U_\alpha = \chi_{D(0,\alpha^{-1})}$, where $\alpha>0$, $D(0,\alpha^{-1})\subset \R^2$ denotes a ball in $\R^2$ with radius $1/\alpha$, and $\chi_A$ is the characteristic function of the set $A\subseteq\R^2$. It can be shown that in both cases $U_\alpha\asts f\to f$ as $\alpha\to 0$. These filters and its variants are often used in the context of the FBP algorithm.   

\item \label{ex:problems:gradient}\textsc{Gradient reconstruction:}
Here $U=U_\alpha$ is chosen as a partial derivative of an approximation of the Delta distribution. For example, as $U_\alpha=(U_\alpha^{(1)},U_\alpha^{(2)})$ with $U_\alpha^{(i)} :=\frac{\partial g_\alpha}{\partial_{x_i}}$, $i=1,2$. In this way, one obtains an approximation of the gradient of $\signal$ via \[\nabla_x f = (U_\alpha^{(1)}\astx f,U_\alpha^{(2)}\astx f) =: U_\alpha \astx \signal,\]
where in the last equation above we applied the convolution $\astx$ componentwise. Such approximations of the gradient are for example used inside the well-known Canny edge detection algorithm \cite{canny1986computational}.

\item \textsc{Laplacian reconstruction:}
Analogously to the gradient approximation, $U$ is chosen to be the Laplacian of an approximation to the Delta distribution. A prominent example, is the Laplacian of Gaussian (LoG), i.e., $U_\alpha=\Delta_x g_\alpha$, also known as the  Marr-Hildreth-Operator. This operator is also used for edge detection, corner detection and blob detection, cf. \cite{marr1980theory}.   
\end{enumerate}
\end{examples}

Depending on the problem at hand, there are several different ways of choosing the regularizer $r(h)$. Prominent examples in the case of image reconstruction include total variation (TV) or the $\ell^1$-norm (possibly in the context of some basis of frame expansion). For the reconstruction of the derivatives (or edges in general), we will use the $\ell^1$-norm as regularization term because derivatives of images can be assumed to be sparse and because the problem \eqref{eq:feature-rec} can be efficiently solved in this case.

\subsection{Filter design}\label{sec:filter}
  
The first step in our framework is filter design for \eqref{eq:feature-rec}. That is, given a feature extraction kernel $U$, we first need to calculate the corresponding filter $u_\Theta = \Ro_\Theta U$ for the CT data, cf.  \eqref{eq:conv}. In our setting, filter design therefore amounts to the evaluation of the Radon transform of $U$. In contrast to our approach, the filter design step in \cite{louis2} consists in the calculation of a solution of the dual equation $U = \Ro^* u$, given the feature extraction filter $U$. As discussed above, the latter case requires full data and might be computationally more involved. From this perspective, filter design required by our approach offers more flexibility and can be considered somewhat simpler.

We now discuss some of the Examples \ref{ex:problems} in more detail and calculate the associated CT data filters $u_\Theta$. In particular, we focus on the Gaussian approximations of the Delta distributions stated in \eqref{eq:gaussian2}. In a first step we compute the Radon transform of a Gaussian.

\begin{Lemma}
\label{lem:radon gaussian}
	The Radon transform of the Gaussian $g_{\alpha}$, defined by \eqref{eq:gaussian2}, is given by
	\begin{equation}
		\Ro g_{\alpha}(\varphi,s) = \frac{1}{\alpha\sqrt{2\pi}}\cdot\exp\left(-\frac{s^{2}}{2\alpha^{2}}\right).
	\end{equation}
\end{Lemma}

Since the Gaussian $g_\alpha$ converges to the Delta distribution as $\alpha\to 0$, the smoothed version $f_\alpha\coloneqq f\astx g_\alpha$ constitutes an approximation to $f$ for small values of $\alpha$. In order obtain approximations to partial derivatives of $f$, we note that \(\frac{\partial f_{\alpha}}{\partial x_{i}} = f\astx \frac{\partial g_{\alpha}}{\partial x_{i}} \). Hence, using the feature extraction filters $U_\alpha^{(i)}:= \frac{\partial g_{\alpha}}{\partial x_{j}}$, the Problem \ref{pr:feature} amounts to reconstructing partial derivatives of $f$. Using this observation together with Lemma \ref{lem:radon gaussian} and the property \ref{r4}, we can explicitly calculate data filters used in different edge reconstruction algorithms (such as Canny or for the Marr-Hildreth-Operator).

\begin{Proposition}
\label{prop:data derivative}
	Let the Gaussian $g_{\alpha}$ be defined by \eqref{eq:gaussian2}. 
	\begin{enumerate}[itemsep=1ex]
	    \item \textsc{Gradient reconstruction:} For the feature extraction filter \(U_{\rm grad}\coloneqq \nabla_x g_\alpha\) the corresponding data filter  \(u_{\rm grad} = (u_\alpha^{(1)},u_\alpha^{(2)}) \) is given by 
	    \begin{equation}
	    \label{eq:radon derivative filter}
		u_{\rm grad}(\varphi,s) = \Ro U_{\rm grad}(\varphi,s)= - \frac{s}{\alpha^{3}\sqrt{2\pi}} \cdot \exp\left(-\frac{s^{2}}{2\alpha^{2}}\right)\cdot \theta(\varphi)
	\end{equation}
	Note that in \eqref{eq:radon derivative filter}, the notation $\Ro U_{\rm grad}$ refers to a vector valued function that is defined by a componentwise application of the Radon transform (cf. Example \ref{ex:problems}, No. \ref{ex:problems:gradient}).
	\item \textsc{Laplacian reconstruction:} For the feature extraction filter $U_\alpha\coloneqq \Delta_x g_\alpha$ the corresponding data filter is given by
	 \begin{equation}
	    \label{eq:radon log filter}
		u_{\rm LoG}(\varphi,s) = \Ro U_{\rm LoG}(\varphi,s)=\frac{1}{\alpha^{3}\sqrt{2\pi}} \cdot \exp\left(-\frac{s^{2}}{2\alpha^{2}}\right)\cdot \left(\frac{s^2}{\alpha^2}-1\right)
	\end{equation}
	\end{enumerate}
\end{Proposition}

From Proposition \ref{prop:data derivative} we immediately obtain an explicit reconstruction formula for the approximate computation of the gradient and of the Laplacian of $f\in\sr(\R^2)$:
\[\nabla_x \signal_\alpha =\Ro^{-1}(u_{\rm grad}  \asts \Ro \signal)\quad\text{ bzw. }\Delta_x\signal_\alpha =\Ro^{-1}(u_{\rm LoG}  \asts \Ro \signal).\]

Both of the above formulas are of FBP type and can be implemented using the standard implementations of the FBP algorithm with a modified filter. This approach has the advantage that only one data filtering step has to performed, followed by the standard backprojection operation.

In order to derive FBP-filters for the gradient and Laplacian reconstruction, let us first note that $\Ro^{-1} = \Ro^\ast\circ\Lambda$, where the operator $\Lambda$ acts on the second variable and is defined in the Fourier domain by $\Fo_s (\Lambda g) (\varphi,\omega) = (4\pi)^{-1}\cdot \abs{\omega}\cdot (\Fo_s g) (\varphi,\omega)$ for $g\in\sr(\sph^1\times\R)$, cf. \cite{natterer}. Now, using the relations for the Fourier transform in 1D,  $\Fo (\mathrm{d} f/ \mathrm{d x})(\omega) = i\cdot  \omega \cdot \hat f(\omega)$,  $\Fo (\mathrm{d}^2 f/ \mathrm{d x}^2)(\omega) = -\omega^2\hat f(\omega)$ and $\Fo(f\ast g)=\sqrt{2\pi}\cdot \hat f\cdot \hat g$, together with
\begin{equation}
    \label{eq:fourier of radon of gaussian}
    \Fo_s(\Ro g_\alpha)(\varphi,s) = \frac{1}{\sqrt{2\pi}}\cdot\exp\left(-\frac{\alpha^2s^2}{2}\right),
\end{equation}
we obtain the following result.

\begin{Proposition}
\label{prop:data derivative 2}
    Let the FBP-filters $W_{\rm grad}=W_{\rm grad}(\varphi,s)$ and $W_{\rm LoG}=W_{\rm LoG}(\varphi,s)$ be given in the Fourier domain (componentwise) by
    \begin{equation}
    \label{eq:fbp filter grad}
        \Fo_s W_{\rm grad}(\varphi,\omega) = \frac{1}{4\pi}\cdot i\cdot \omega\cdot \abs{\omega}\cdot\exp\left(-\frac{-\alpha^2s^2}{2}\right)\cdot \theta(\varphi),
    \end{equation}
    and
    \begin{equation}
    \label{eq:fbp filter laplacian}
        \Fo_s W_{\rm grad}(\varphi,\omega) = - \frac{1}{4\pi}\cdot \abs{\omega}^3\cdot\exp\left(-\frac{-\alpha^2s^2}{2}\right),
    \end{equation}
    where $\varphi\in(0,2\pi)$ and $\omega\in\R$. Then, for $f\in\sr(\R^2)$, we have
    \begin{equation}
    \label{eq:fbp rec formulae}
        \nabla_x \signal_\alpha =\Ro^{\ast}(W_{\rm grad}  \asts \Ro \signal)\quad\text{ and }\quad \Delta_x\signal_\alpha =\Ro^{\ast}(W_{\rm LoG}  \asts \Ro \signal).
    \end{equation}
\end{Proposition}

Since the FBP algorithm is a regularized implementation of $\Ro^{-1}$ (cf. \cite{natterer}), a standard toolbox implementation could be used in practice in order to compute $\nabla_x f_{\alpha}$ and $\Delta_x f$. To this end, one only needs to use the modified filters for the FBP, provided in \eqref{eq:fbp filter grad} and \eqref{eq:fbp filter laplacian}, instead of the standard FBP filter (such as Ram-Lak). Again, let us emphasize that the reconstruction formulae \eqref{eq:fbp rec formulae} can only be used in the case of properly sampled CT data. If the CT data does not satisfy the sampling requirements, e.g., in case of angular undersampling, this FBP algorithm will produce artifacts which can substantially degrade the performance of edge detection. In such cases, our framework \eqref{eq:feature-rec} should be used in combination with a suitable regularization term. In the context of edge reconstruction, we propose to use $\ell^1$-regularization in combination with $\ell^2$-regularization. This approach will be discussed in the next section.

So far, we constructed data filters for the approximation of the gradient and Laplacian in the spatial, cf. Proposition \ref{prop:data derivative}, and derived according FBP filters in the Fourier domain in Proposition \ref{prop:data derivative 2}. In a similar fashion, one can derive various related  examples  by replacing the Gaussian by feature kernels whose Radon transform is known analytically. Another way of obtaining practically relevant data filters (for a wide class of feature filter) is to derive expressions for the data filters in the Fourier domain (i.e., filter design in the Fourier domain). In the following, we provide two basic examples for filter design in the Fourier domain. To this end, we will employ the Fourier slice theorem, cf. Lemma~\ref{lem:radon}, \ref{r1}.

\begin{remark}
\label{rem:data filters}
\item[]
    \begin{enumerate}[itemsep=1ex]
        \item \label{ex:lowpass} \textsc{Lowpass Lalpacian:} The Laplacian of the ideal lowpass is defined as \[U_b = \Delta_x \Fo^{-1} (\chi_{D(0,b)}),\] where $b$ is the bandwidth of $U_b$. Using the property \ref{r5}, we get $\Ro(U_b) = \frac{\partial^2}{\partial s^2}\Ro(\Fo^{-1}(\chi_{D(0,b)}))$. By the Fourier slice theorem, we obtain \[\Fo_s(\Ro(U_b))(\varphi,\omega) = -\omega^2 \chi_{D(0,b)}(\omega\cdot\theta(\varphi)) = -\omega^2 \chi_{[-b,b]}(\omega).\]
        Hence, the associated data filter is given by
        \begin{align}
            \notag
            u_b(\varphi,s):=\Ro U_b(\varphi,s) &= \frac{\partial^2}{\partial s^2} \Fo_s^{-1}(\chi_{[-b,b]})(s) = \sqrt{\frac{2}{\pi}}\cdot \frac{\partial^2}{\partial s^2} \frac{\sin(bs)}{s}\\[1ex]
            \label{filterchi}
            &= \sqrt{\frac{2}{\pi}}\cdot \left(\frac{2\sin(bs)}{s^3}-\frac{2b\cos(bs)}{s^2}-\frac{b^2\sin(bs)}{s}\right)
        \end{align}
       
        Because $u_b$ is $b$-band-limited, the convolution with the filter~\eqref{filterchi} can  be discretized systematically whenever  the underlying image is essentially $b$ band-limited. To that end, assume that the function $\signal$ has  bandwidth $b$. Then $\data = \Ro \signal$ has bandwidth $b$ as well (with respect to the second variable), and therefore the continuous convolution $\Ro \signal \asts u_b$  can be exactly computed via discrete convolution. Using  discretization \eqref{eq:sampling} and taking $s_\ell = \frac{\pi}{b}\cdot \ell $ we obtain from  \eqref{filterchi} the discrete filter
        \begin{align}\label{filterquadratic}
        u_b(\varphi,s_\ell)=
        - \sqrt{\frac{2}{\pi}}\cdot b^3\cdot \begin{cases}
              \dfrac{1}{3},  & \text{ if}\quad \ell=0, \\[2ex]
             \dfrac{2\cdot  (-1)^\ell}{\pi^2 \ell^2}  & \text{ if}\quad \ell\neq 0.
           \end{cases}
        \end{align}
        According to one-dimensional Shannon sampling theory, we compute $\data \asts u_b $ via discrete convolution with the  filter coefficients given in \eqref{filterquadratic}.       
        
        \item\label{ex:ramlak} \textsc{Ram-Lak type filter:} Consider the feature extraction filter \[ U_{b,1} = \Delta_x  \, \Fo^{-1} \left[\chi_{D(0,b)} \cdot (1-\norm{\,\cdot\,})_+ \right], \] where $(1-\norm{\,\cdot\,})_+\coloneqq\max\set{0,1-\norm{\,\cdot\,}}$. Note that for $b\geq 1$, we have $u_{b,1}=u_{1,1}$, since in this case $\chi_{D(0,b)} \cdot (1-\norm{\,\cdot\,})_+ = (1-\norm{\,\cdot\,})_+$. Hence, we consider the case $b\leq 1$. In a similar fashion as above,  we obtain
         \begin{align}
            \notag
            u_{b,1}:=\Ro U_{b,1}(\varphi,s) &= \frac{\partial^2}{\partial s^2} \Fo_s^{-1}\left[\chi_{[-b,b]}\cdot (1-\abs{\,\cdot\,})\right] (s)\\[1ex]
            &= \frac{\partial^2}{\partial s^2} \left[\Fo_s^{-1}[\chi_{[-b,b]}] (s) - \Fo_s^{-1} [\abs{\,\cdot\,}\cdot\chi_{[-b,b]}](s) \right] .
        \end{align}
        
        Evaluating $u_{b,1}$ at  $s_\ell = \frac{\pi}{b}\cdot \ell $, we get
        \begin{align}\label{filterramlaktype}
        u_{\rm b,1} (\theta, s_\ell) = \sqrt{\frac{2}{\pi}}\cdot b^3\cdot
        \begin{cases}
              \phantom{-}\dfrac{3b-4}{12}  & \text{ if } \ell=0 \\[2ex]
                \phantom{-}\dfrac{3b-2}{\pi^2 \ell^2}   & \text{ if } \ell  \text{ is even} \\[2ex]
                -\dfrac{3b-2}{\pi^2 \ell^2}+\dfrac{12b}{\pi^4\ell^4} & \text{ if } \ell \text{ is odd }.
           \end{cases}
        \end{align}
        Again, we can evaluate  $\data \asts u_{b,1} $ via discrete convolution with the filter coefficients  \eqref{filterramlaktype}.
    \end{enumerate}
\end{remark}

Finally, let us note that there are several other examples for feature reconstruction filters for which one can derive explicit formulae of corresponding data filters in a similar way as we did in this section. For example, in the case of approximation of Gaussian derivatives of higher order or for band-limited versions of derivatives.

\section{Numerical results}\label{sec:num}
 
In our the numerical experiments, we focus on the reconstruction of edge maps. To that end, we use our framework \eqref{eq:feature-rec} in combination with feature extraction filters that we have derived in Proposition \ref{prop:data derivative} and in Remark \ref{rem:data filters}. Since the gradient and the Laplacian of an image have relatively large values only around edges and small values elsewhere, we aim at exploiting this sparsity and, hence, use a linear combination $ \reg(h) = \mu \norm{ \nabla h }_2^2 + \lambda \norm{h}_1$ as regularizer in \eqref{eq:feature-rec}. The resulting minimization problem then reads    
\begin{equation}\label{eq:minnoisy2}
	\frac{1}{2} \, \norm{\Ro_\Theta h - u_\Theta  \asts  \data_\Theta }_2^2
	+ 
	 \mu \Vert \nabla h\Vert_2^2 + \lambda \norm{h}_1
	\to 
	\min_{h \in \XX_0}  \,.
\end{equation}
If $\mu=0$, this approach reduces to the $\ell^1$ regularization which is known to favor sparse solution. If $\mu\neq 0$, the additional $H^1$-term  increases smoothness of the recovered edges. In order to numerically minimize  \eqref{eq:minnoisy2}, we use the fast iterative shrinkage-thresholding algorithm (FISTA) of  \cite{fista}. Here, we apply the  forward step  to $\frac{1}{2}\norm{\Rd_\Theta h - u_\Theta  \asts  \data_\Theta }_2^2	+  \mu \Vert \nabla h\Vert_2^2$ and the backward step to $\lambda \norm{h}_1$. The discrete $\ell^p$ norms are  defined  by $\norm{h}_p = (\sum_{i,j=1}^N \abs{h_{ij}}^p)^\frac{1}{p}$ and the  discrete Radon transform $\Rd_\Theta$ is computed via the composite trapezoidal rule and bilinear interpolation. The adjoint Radon transform $\Rd_\Theta^\ast$ is implemented as a discrete backprojection following \cite{natterer}.

\subsection{Reconstruction of the Laplacian feature map}
\label{subsec:rec of laplacian}
We first investigate the feasibility of the proposed approach for recovering the Laplacian of the initial image. For our first experiment, we use a phantom image, which is defined as a characteristic function of the union of three discs and has the size $N\times N$ with $N=200$, cf. Figure \ref{subfig:phantom}. Since, according to  the sampling condition \eqref{eq:sampling-cond}, full aliasing free angular sampling requires $\lceil \pi N_s \rceil = 472$ samples in the $s$-variable, we computed tomographic data at $2N_s +1= 301$ equally spaced signed distances $s_\ell \in [-1.5,1.5]$ and at $N_\varphi=40$ equally spaced directions in $[0, \pi)$. This data is properly sampled in the $s$-variable, but undersampled in the angular variable $\varphi$, cf. Figure~\ref{subfig:data}. 

From this data, we computed the approximate Laplacian reconstruction, shown in Figure \ref{subfig:fbp-log}, using the standard FBP algorithm in combination with the LoG-filtered data $u_{\rm LoG}\asts y_\Theta$ that we computed in a preprocessing step using the LoG data filter from Proposition \ref{prop:data derivative}. It can be clearly observed that FBP introduces prominent undersampling artefacts (streaks), so that many edges in the calculated feature map are not related to the actual image features. This shows, that the edge maps computed by FBP (from undersampled data) can include unreliable information and even falsify the true edge information (since artefacts and actual edges superimpose). In a more realistic setup, this could be even worse, since artefacts may not be that clearly distinguishable from actual edges.

Figure~\ref{fig:lap} shows reconstructions of feature maps from  noise-free CT data that we computed using our framework \eqref{eq:minnoisy2} for three different choices of feature extraction filters and for two different sets of regularization parameters. The first row of Figure~\ref{fig:lap} shows reconstructions with $\mu=0$ and $\lambda=0.001$ using 1000 iterations of the FISTA algorithm, whereas the second row shows reconstructions that were computed using an additional $H^1$-term with $\lambda=\mu=0.001$ and using 500 iterations of the FISTA algorithm. In contrast to the FBP-LoG reconstruction (shown in Figure \ref{subfig:fbp-log}), the undersampling artefacts have been removed in all cases. As expected, the use of $\ell^1$-regularization without an additional $H^1$-smoothing (shown in first row) produces sparser feature maps, as opposed to reconstruction shown in the second row. However, we also observed, that the iterative reconstruction based only on the $\ell^1$-minimization (without the $H^1$-term) sometimes has trouble reconstructing the object boundaries properly. In fact, we found that a proper reconstruction of boundaries is quite sensitive to the choice of the $\ell^1$-regularization parameter. If this parameter is chosen too large, we observed that the boundaries could be incomplete or even disappear. Since the $\ell^1$-regularization parameter controls the sparsity of the reconstructed feature map, this observation is actually not surprising. By including an additional $H^1$-regularization term, the reconstruction results become less sensitive to the choice of regularization parameters.

\begin{figure}[htb]
\begin{subfigure}{.32\textwidth}
  \centering
  \includegraphics[width=\textwidth,height=\textwidth]{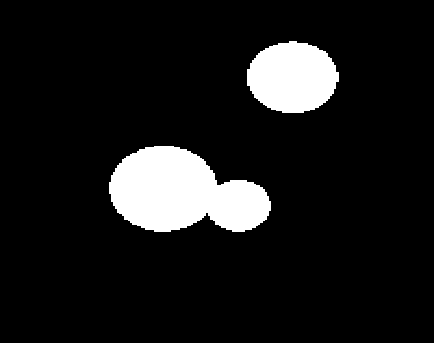}
  \caption{Phantom}
  \label{subfig:phantom}
\end{subfigure}
\hfill
\begin{subfigure}{.32\textwidth}
    \centering
    \includegraphics[width=\textwidth,height=\textwidth]{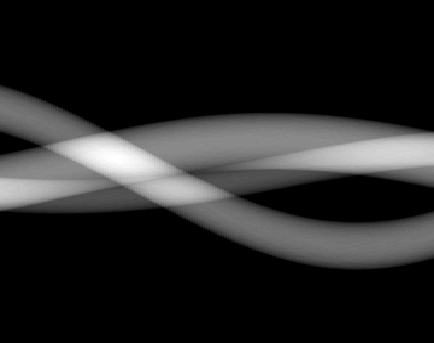}
    \caption{CT data}
  \label{subfig:data}
\end{subfigure}
\hfill
\begin{subfigure}{.32\textwidth}
    \centering
    \includegraphics[width=\textwidth,height=\textwidth]{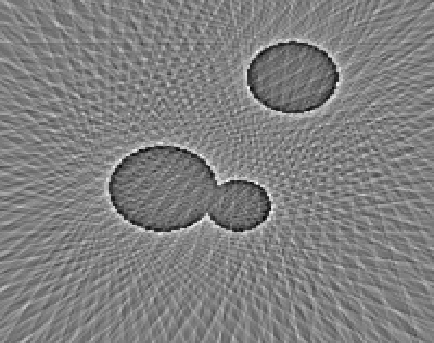}
    \caption{FBP-LoG reconstruction}
  \label{subfig:fbp-log}
\end{subfigure}
\medskip

\caption{\textsc{Reconstruction of the Laplacian feature map using FBP}. 
The phantom image of size $200\times200$ consisting of a union of three discs (\subref{subfig:phantom}) and the corresponding angularly undersampled CT data, measured at 40 equispaced angles in $[0,\pi)$ and properly sampled in the $s$-variable with $301$ equispaced samples $s_\ell \in[-1.5,1.5]$ (\subref{subfig:data}). Subfigure (\subref{subfig:fbp-log}) shows the Laplacian of Gaussian (LoG) reconstruction using the standard FBP algorithm. It can be clearly observed that FBP introduces prominent streaking artefacts that are due to the angular undersampling.}
\end{figure}

\vskip3ex

\begin{figure}[htb!]
\begin{subfigure}{.32\textwidth}
  \centering
  \includegraphics[width=\textwidth,height=\textwidth]{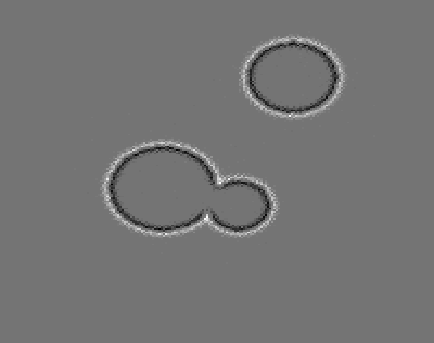}
  \caption{LoG: $\mu=0$, $\lambda=0.001$}
  \label{subfig:log}
\end{subfigure}
\hfill
\begin{subfigure}{.32\textwidth}
    \centering
    \includegraphics[width=\textwidth,height=\textwidth]{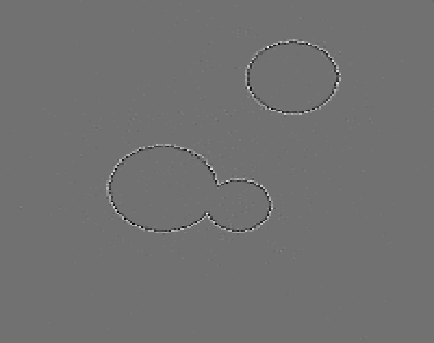}
    \caption{Lowpass:$\mu=0, \lambda=0.001$}
  \label{subfig:lowpass}
\end{subfigure}
\hfill
\begin{subfigure}{.32\textwidth}
    \centering
    \includegraphics[width=\textwidth,height=\textwidth]{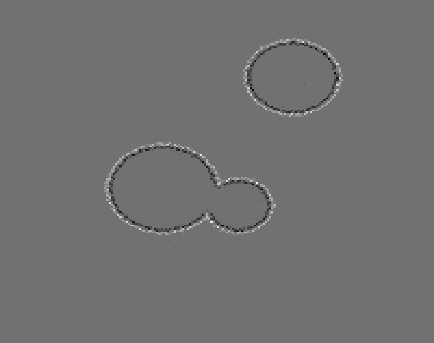}
    \caption{Ram-Lak:$\mu=0, \lambda=0.001$}
  \label{subfig:ramlak}
\end{subfigure}
\vskip3ex
\begin{subfigure}{.32\textwidth}
  \centering
  \includegraphics[width=\textwidth,height=\textwidth]{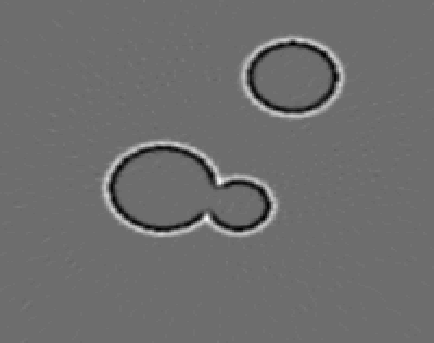}
  \caption{LoG: $\mu=\lambda=0.001$}
  \label{subfig:logH1}
\end{subfigure}
\hfill
\begin{subfigure}{.32\textwidth}
    \centering
    \includegraphics[width=\textwidth,height=\textwidth]{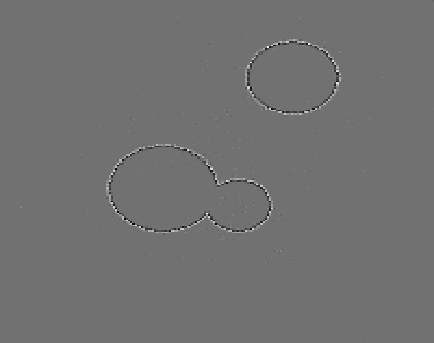}
    \caption{Lowpass: $\mu=\lambda=0.001$}
  \label{subfig:lowpassH1}
\end{subfigure}
\hfill
\begin{subfigure}{.32\textwidth}
    \centering
    \includegraphics[width=\textwidth,height=\textwidth]{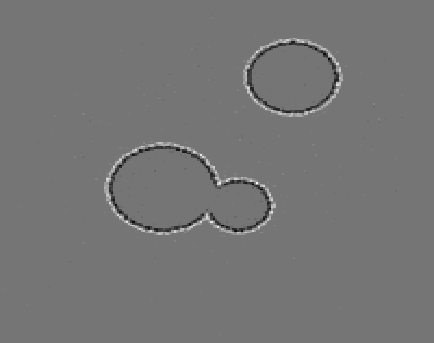}
    \caption{Ram-Lak: $\mu=\lambda=0.001$}
  \label{subfig:ramlakH1}
\end{subfigure}
\medskip

\caption{\textsc{Reconstruction of Laplacian feature maps using our framework}. This figure shows reconstructions of feature maps from noise-free CT data that we computed using our framework \eqref{eq:minnoisy2} for three different choices of feature extraction filters and for two different sets of regularization parameters. Here, LoG refers to  \eqref{eq:radon log filter}, Lowpass to \eqref{filterquadratic}, and Ram-Lak to \eqref{filterramlaktype}.  The first row shows reconstructions with $\mu=0$ and $\lambda=0.001$ using 1000 iterations of the FISTA algorithm, whereas the second row shows reconstructions that were computed using an additional $H^1$-term with $\lambda=\mu=0.001$ and using 500 iterations of the FISTA algorithm. In contrast to the FBP-LoG reconstruction (shown in Figure \ref{subfig:fbp-log}), the undersampling artefacts have been removed in all cases.
}
\label{fig:lap}
\end{figure}

\begin{figure}[htb!]
\begin{subfigure}{.32\textwidth}
  \centering
  \includegraphics[width=\textwidth,height=\textwidth]{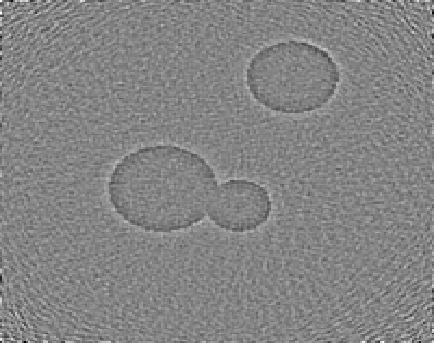}
  \caption{Ram-Lak: $\mu=0.001, \lambda=0$}
  \label{subfig:noisy ramlak1}
\end{subfigure}
\hfill
\begin{subfigure}{.32\textwidth}
    \centering
    \includegraphics[width=\textwidth,height=\textwidth]{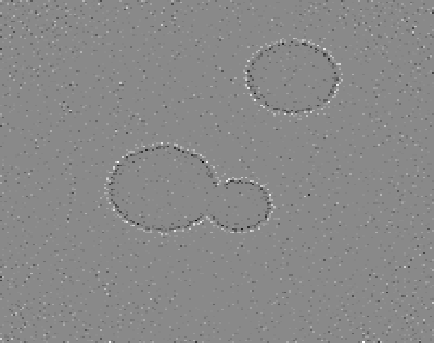}
    \caption{Ram-Lak: $\mu=0, \lambda=0.001$}
  \label{subfig:noisy ramlak2}
\end{subfigure}
\hfill
\begin{subfigure}{.32\textwidth}
    \centering
    \includegraphics[width=\textwidth,height=\textwidth]{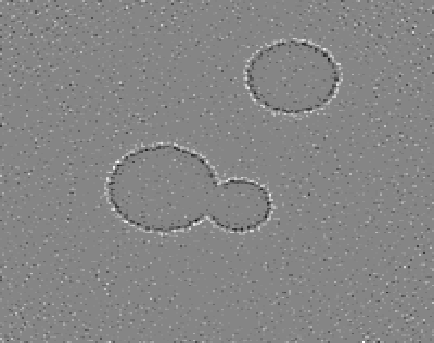}
    \caption{Ram-Lak: $\mu=\lambda=0.001$}
  \label{subfig:noisy ramlak3}
\end{subfigure}
\medskip
\caption{\textsc{Reconstructions of Laplacian feature maps from noisy data.} Reconstruction in (\subref{subfig:noisy ramlak1}) were calculated using only $H^1$-regularization, in (\subref{subfig:noisy ramlak2}) using only $\ell^1$-regularization, and in (\subref{subfig:noisy ramlak3}) using  combined $\ell^1$ and $H^1$-regularization.}\label{fig:ergl1H1noise}
\end{figure}

In order to simulate the real world measurements  more realistically, we added  noise to the CT data that we used in the previous experiment. Using this noisy data, we calculated reconstructions via \eqref{eq:minnoisy2} in combination with the Ram-Lak type filter \eqref{filterramlaktype} using three different sets of regularization parameters and 1000 iterations of the FISTA algorithm in each case. The reconstruction using the parameters $\lambda=0$ and $\mu=0.001$ (i.e., only $H^1$-regularization was applied) is shown in Figure \ref{subfig:noisy ramlak1}. The reconstruction in Figure \ref{subfig:noisy ramlak2} uses only $\ell^1$ regularization, i.e., $\mu=0$ and $\lambda=0.001$, and the  reconstruction in Figure \ref{subfig:noisy ramlak3} applies both regularization terms with $\lambda = \mu=0.001$. In both reconstructions shown in Figure \ref{subfig:noisy ramlak2}  and \ref{subfig:noisy ramlak3}, the recovered features are much more apparent than for pure $H^1$-regularization. As in the noise-free situation, we observe that the (pure) $\ell^1$-regularization might generate discontinuous  boundaries, whereas the combined $H^1$-$\ell^1$-regularization  results in smoother and (seemingly) better represented edges.

\subsection{Edge detection}

One main application of our framework for the reconstruction of approximate image gradients or approximate Laplacian feature maps is in edge detection. Clearly, feature maps that contain less artifacts can be expected to provide more accurate edge maps. 

For this experiment, we used a modified phantom image that is shown in Figure \ref{subfig:mod phantom}. In contrast to the previously used phantom, this image includes also weaker edges that are more challenging to detect. For this phantom, we generated CT data using the same sampling scheme as in our first experiment (Section \ref{subsec:rec of laplacian}) and computed the LoG-feature maps $\signal  \astx  U_{\rm LoG}$ using the FBP approach (cf. Figure \ref{subfig:fbp-log mod}) and using our approach (cf. Figure \ref{subfig:log mod}) with $\mu=0$, $\lambda  = 0.002$, and $100$ iterations of the FISTA algorithm for \eqref{eq:minnoisy2}. Subsequently, we generated corresponding binary edge maps by extracting the zero crossings of these LoG-feature maps (cf. Figure \ref{subfig:fbp-log edge map} and \ref{subfig:log edgemap}) by using the MATLABs \texttt{edge} functions. Note that this procedure is one of the standard edge detection algorithm, known as the LoG edge detector, cf.  \cite{marr1980theory}. For both methods, we took standard deviation $\alpha = 1.3$ for the application of the Gaussian smoothing  and  threshold $t = 0.005$ for the detection of the zero crossings. As can be clearly seen from the results, the edge detection based on our approach (cf. Figure \ref{subfig:fbp-log edge map}) is able to detect also the weaker edges inside the large disc. In contrast, edge detection in combination with FBP-LoG featuremap was not able to detect the edge set correctly due to strong undersampling artefacts.

\begin{figure}[htb!]
\begin{subfigure}{.32\textwidth}
  \centering
  \includegraphics[width=\textwidth,height=\textwidth]{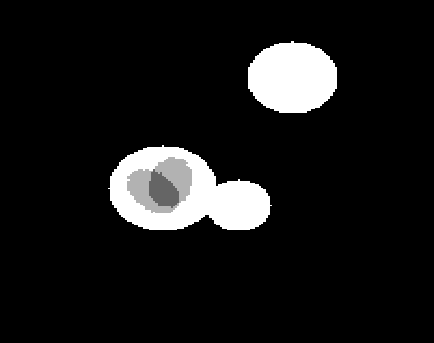}
  \caption{Modified phantom}
  \label{subfig:mod phantom}
\end{subfigure}
\hfill
\begin{subfigure}{.32\textwidth}
    \centering
    \includegraphics[width=\textwidth,height=\textwidth]{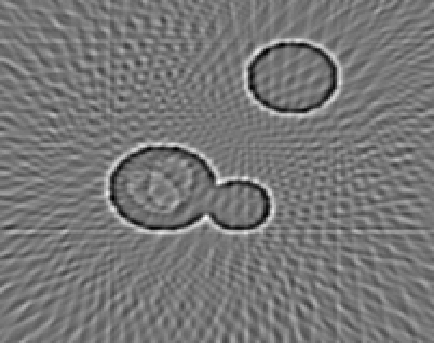}
    \caption{FBP-LoG}
  \label{subfig:fbp-log mod}
\end{subfigure}
\hfill
\begin{subfigure}{.32\textwidth}
    \centering
    \includegraphics[width=\textwidth,height=\textwidth]{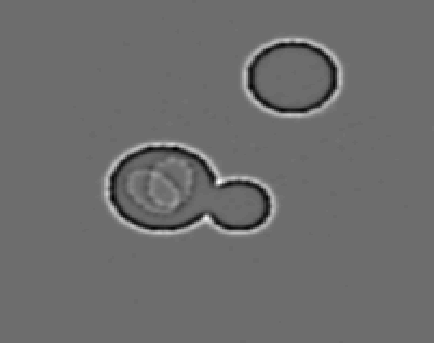}
    \caption{LoG: $\mu=\lambda=0.001$}
  \label{subfig:log mod}
\end{subfigure}
\vskip3ex
\phantom{
\begin{subfigure}{.3\textwidth}
  \centering
  \includegraphics[width=\textwidth,height=\textwidth]{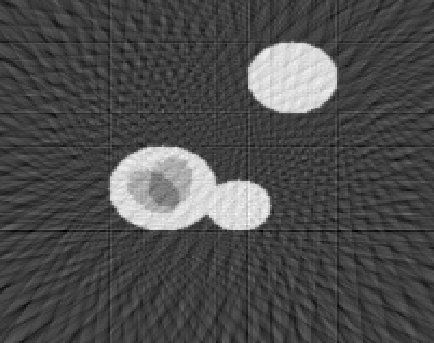}
\end{subfigure}
}
\hfill
\hspace{0.12cm}
\begin{subfigure}{.32\textwidth}
    \centering
    \includegraphics[width=\textwidth,height=\textwidth]{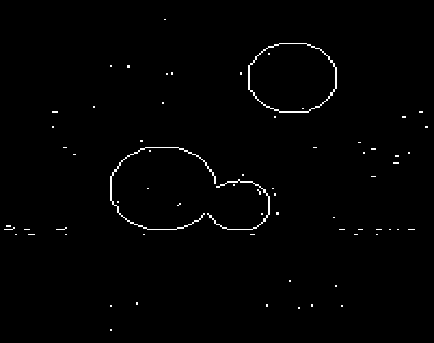}
    \caption{FBP-LoG: edge map}
  \label{subfig:fbp-log edge map}
\end{subfigure}
\hfill
\hspace{0.025cm}
\begin{subfigure}{.32\textwidth}
    \centering
    \includegraphics[width=\textwidth,height=\textwidth]{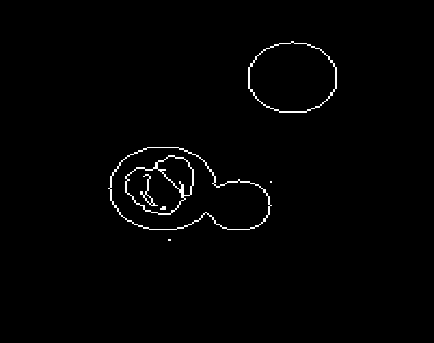}
    \caption{LoG: edge map}
  \label{subfig:log edgemap}
\end{subfigure}
\medskip

\caption{\textsc{LoG edge detection}. 
The modified phantom image (\subref{subfig:mod phantom}) includes also weaker edges that are more challenging to detect. Subfigures (\ref{subfig:fbp-log mod}) and (\ref{subfig:log mod}) show reconstructions of the LoG feature maps, that were generated using the FBP algorithm and our approach, respectively. The corresponding binary edge masks generated by the LoG edge detector are shown in (\ref{subfig:fbp-log edge map}) and (\ref{subfig:log edgemap}). }\label{fig:edge}
\end{figure}

\begin{figure}[htb!]
\begin{subfigure}{.32\textwidth}
  \centering
  \includegraphics[width=\textwidth,height=\textwidth]{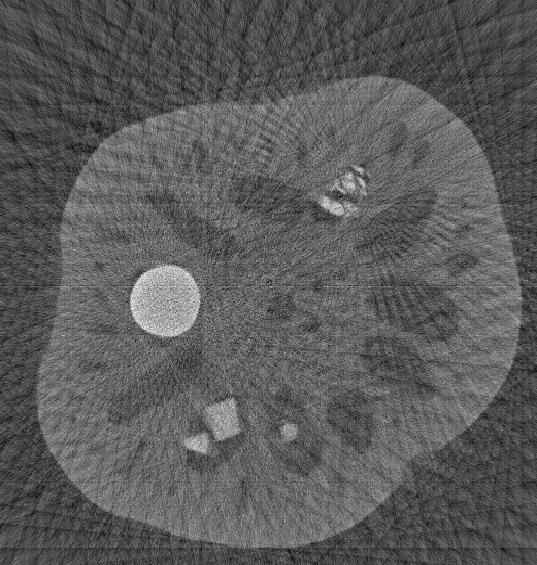}
  \caption{FBP reconstruction}
  \label{subfig:lotus fbp}
\end{subfigure}
\hfill
\begin{subfigure}{.32\textwidth}
    \centering
    \includegraphics[width=\textwidth,height=\textwidth]{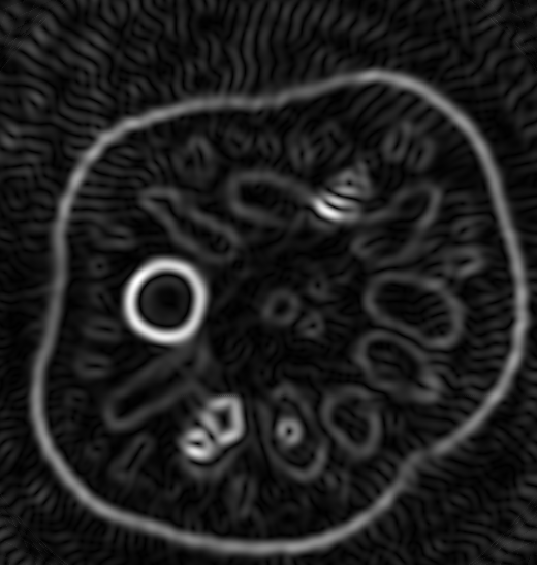}
    \caption{FBP-grad}
  \label{subfig:lotus fbp gradient}
\end{subfigure}
\hfill
\begin{subfigure}{.32\textwidth}
    \centering
    \includegraphics[width=\textwidth,height=\textwidth]{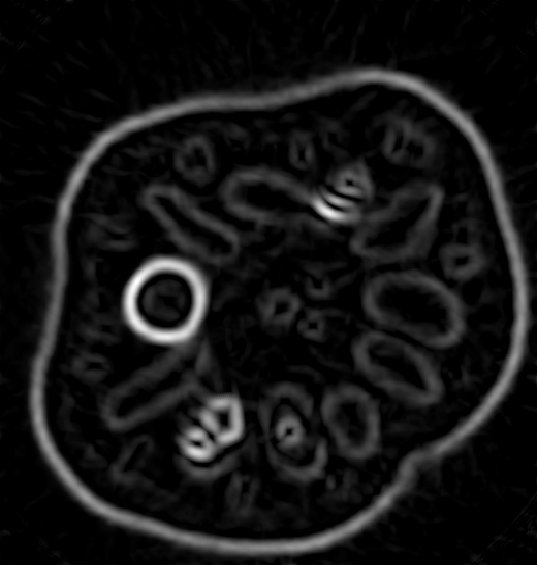}
    \caption{grad: $\mu=0$, $\lambda=0.01$}
  \label{subfig:lotus log gradient}
\end{subfigure}
\vskip3ex
\begin{subfigure}{.32\textwidth}
  \centering
  \includegraphics[width=\textwidth,height=\textwidth]{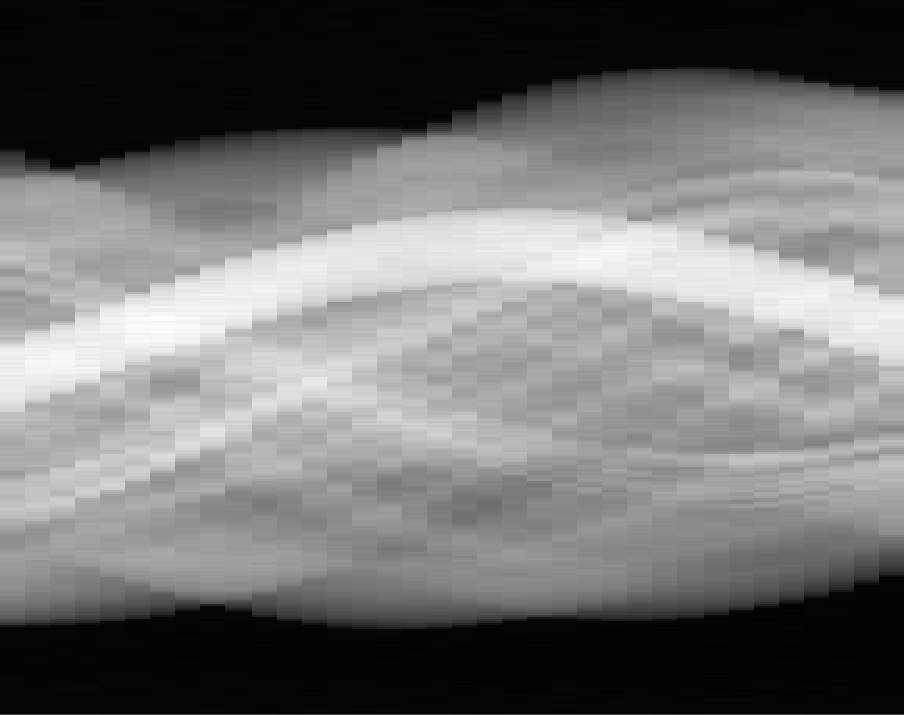}
  \caption{CT data of a lotus }
  \label{subfig:lotus sino}
\end{subfigure}
\hfill
\hspace{0.025cm}
\begin{subfigure}{.32\textwidth}
    \centering
    \includegraphics[width=\textwidth,height=\textwidth]{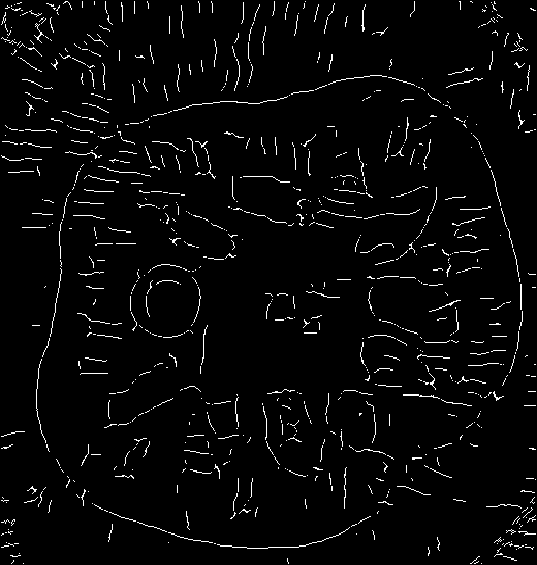}
    \caption{FBP-grad: edge map}
  \label{subfig:lotus fbp edge}
\end{subfigure}
\hfill
\hspace{0.025cm}
\begin{subfigure}{.32\textwidth}
    \centering
    \includegraphics[width=\textwidth,height=\textwidth]{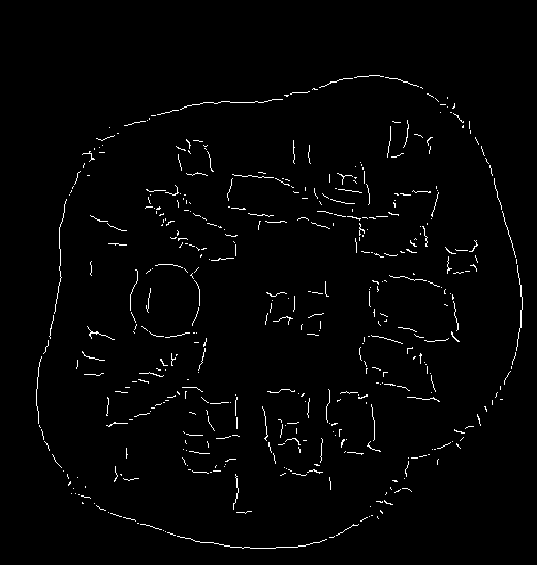}
    \caption{grad: edge map}
  \label{subfig:lotus log edge}
\end{subfigure}
\medskip

\caption{\textsc{Canny edge detection from the lotus data set.}
Rebinned CT data of a lotus root (\subref{subfig:lotus sino}) (cf. \cite{bubba2016tomographic}) and the corresponding FBP reconstruction (\subref{subfig:lotus fbp}) from  $36$ evenly distributed angles in $[0,\pi)$. Magnitude of the smooth gradient map $\abs{\nabla  U_{\rm grad} \astx \signal}$  computed using the FBP algorithm (\subref{subfig:lotus fbp gradient}) and using our approach (\subref{subfig:lotus log gradient}). The corresponding edge detection results using the Canny algorithm  are shown in (\subref{subfig:lotus fbp edge}) and (\subref{subfig:lotus log edge}), respectively.}\label{fig:lotus}
\end{figure}

In our last experiment we present edge detection results for real  CT data of a lotus \cite{bubba2016tomographic}. Note that similar reconstructions were presented in \cite{frikel2021combining}. The lotus data have been rebinned and downsampled to $2 N_s +1= 739$ signed distances and $N_\ph  = 36$ directions. The Gaussian gradient feature map was computed in two ways: First, by applying FBP to filtered CT data with the data filter \eqref{eq:radon derivative filter}, cf. Figure \ref{subfig:lotus fbp gradient}. Second, by using our approach \eqref{eq:feature-rec} with $\mu=0$ and $\lambda=0.01$, and by applying 50 iterations of the FISTA algorithm, cf. Figure \ref{subfig:lotus log gradient}. The resulting image size is $521 \times 521$. The standard deviation for the Gaussian smoothing was chosen as $\alpha = 6$ and for the Canny edge detection we used the  same lower threshold $0.1$ and upper threshold $0.15$. In order to calculate binary edge maps (shown in Figures \ref{subfig:lotus fbp edge} and \ref{subfig:lotus log edge}), we used the Canny edge detector (cf. \cite{canny1986computational}) in combination with the pointwise magnitude of the Gaussian gradient maps $\abs{\nabla U_{\rm grad} \astx f }$. Again, it was observed that the calculation of the Gaussian gradient map using our approach leads to more reliable edge detection results.

 \section{Conclusion}\label{sec:outlook}

In this paper, we proposed  a framework for the  reconstruction of features maps  directly from  incomplete tomographic data, without the need of reconstructing the tomographic image $\signal $ first. Here,  a features map  refers to the convolution $U \astx \signal$ where $U$ is a given convolution kernel and $\signal$ is the underlying object. Starting from the forward convolution identity for the Radon transform, we introduced a variational model for feature reconstruction, that is formulated using the discrepancy term $ \norm{\Ro_\Theta h - u_\Theta  \asts  \data_\Theta }_2^2 $ and a general regularizer $\reg( h )$. In contrast to existing approaches, such as \cite{louis2,louis1}, our framework does not require full data and, due to the variational formulation, also offers a flexible way for integrating a priori information about the feature map into the reconstruction. In several numerical experiments, we have illustrated that our method can outperform classical feature reconstruction schemes, especially, if the CT data is incomplete. Although we mostly focused on reconstruction of feature maps that are used for edge detetion purposes, our framework can be adapted for a wide range of problems. A rigorous convergence analysis of the presented scheme remains an open issue. Another directions of further research may may include the extension of the proposed approach to non-sparse, non-convolutional features and generalization to other types of tomography problems. Also, multiple feature reconstruction (similar to the method \cite{jointmarkus,zangerl2020multi}) seems to be an interesting future research direction.  

\vspace{6pt} 

\section*{Acknowledgment}
The work of the M.H was supported by Austrian Science Fund (FWF) project P 30747-N32. The contribution by S.G is part of a project that has received funding from the European Union’s Horizon 2020 research and innovation programme under the Marie Skłodowska-Curie grant agreement No 847476. The views and opinions expressed herein do not necessarily reflect those of the European Commission.

\bibliography{literatur.bib}

\end{document}